\documentclass[11pt]{article}
\usepackage[margin=1in]{geometry}
\usepackage{amsmath,amsthm,amsfonts,amssymb}
\usepackage{enumerate}
\usepackage[authoryear,round]{natbib}
\usepackage{enumitem}
\usepackage{graphicx}
\usepackage{booktabs}
\usepackage{float}

\PassOptionsToPackage{hyphens}{url}
\usepackage{url}
\usepackage{xcolor}
\usepackage{hyperref}
\hypersetup{
  colorlinks = true,
  linkcolor  = [rgb]{0.10,0.20,0.55},
  citecolor  = [rgb]{0.10,0.20,0.55},
  urlcolor   = [rgb]{0.10,0.20,0.55}
}

\usepackage{mathrsfs}

\allowdisplaybreaks
\newtheorem{theorem}{Theorem}[section]
\newtheorem{corollary}[theorem]{Corollary}
\newtheorem{lemma}[theorem]{Lemma}
\newtheorem{proposition}[theorem]{Proposition}
\theoremstyle{definition}

\newtheorem{assumption}[theorem]{Assumption}
\theoremstyle{remark}
\newtheorem{remark}[theorem]{Remark}

\newcommand{\R}{\mathbb{R}}
\newcommand{\C}{\mathbb{C}}
\newcommand{\E}{\mathbb{E}}
\newcommand{\Q}{\mathbb{Q}}

\newcommand{\F}{\mathcal{F}}

\newcommand{\Rr}{\mathcal{R}}

\newcommand{\mut}{\widetilde\mu}

\newcommand{\dd}{\mathrm{d}}

\newcommand{\cK}{\mathsf{K}}

\title{Variance-Optimal Hedging in the Rough Hawkes--Heston Model}
\author{Yingli Wang\thanks{School of Mathematical Sciences, Fudan University,
  Shanghai, People's Republic of China;
  \texttt{yingliwang@fudan.edu.cn}; numerical code:
  \url{https://github.com/gagawjbytw/hedging_rough_hawkes_heston}}
  \and
  Xiaoyu Wang\thanks{Corresponding Author. FinTech Thrust, Hong Kong University of Science and Technology (Guangzhou), Guangzhou, Guangdong Province, People's Republic of China;
\texttt{xiaoyuwang@hkust-gz.edu.cn}}}
\date{\today}

\begin{document}
\maketitle

\begin{abstract}
We study variance-optimal stock hedging and the convergence of approximate
strategies in the rough Hawkes--Heston model. Starting from the model's affine
conditional transform and the affine Volterra jump framework, we obtain
semi-explicit hedges for European calls and
a representation of the minimum quadratic error through the
Galtchouk--Kunita--Watanabe projection. Our main approximation result keeps
the original stock, variance driver, and information flow fixed while
regularizing the kernel used to evaluate the hedge. To handle singular
memory and common marked jumps, we construct the approximate holdings from
histories available before trading and preserve the conditional transform's
random modulus envelope. Riccati--Volterra stability and weighted truncation
then yield convergence in the original stock's trading norm on compact
Fourier intervals. For calls, a joint choice of kernel regularization and
Fourier cutoff gives convergence of the initial capitals and strategies,
uniform-in-time square-mean convergence of continuous-time gains, and
convergence of the terminal mean-square error to the variance-optimal value.
A numerical experiment with shifted fractional kernels illustrates the
construction on common original-market paths.
\end{abstract}



\section{Introduction}
\label{sec:introduction}

The Heston model is a classical affine stochastic-volatility benchmark with
tractable Markovian dynamics \citep{heston1993}. Continuous-time
long-memory volatility models were developed by \citet{comterenault1998}, with
an affine fractional extension in \citet{comtecoutinrenault2012}. The modern
rough-volatility literature emphasizes local sample-path regularity: realized
log-volatility behaves approximately like a fractional process with Hurst
parameter below one half \citep{gatheral2018}, and option-implied estimates
give complementary evidence of roughness \citep{livierimoutipallavicini2018}.
The power-law behavior of the short-maturity at-the-money skew is linked to
rough fractional volatility by \citet{fukasawa2017,fukasawa2021}.

The rough Bergomi model of \citet{bayerfrizgatheral2016} and the rough Heston
model of \citet{eleuchrosenbaum2019} are two principal model constructions.
Affine forward-variance and stochastic Volterra formulations developed by
\citet{gatheralkellerressel2019,abijaberlarssonpulido2019} explain how
non-Markovian volatility retains exponential-affine transforms; see also the
rough affine synthesis of \citet{kellerressellarssonpulido2023}. Nearly
unstable heavy-tailed Hawkes processes provide a complementary
microstructural route to fractional volatility and leverage
\citep{jaissonrosenbaum2016,eleuchfukasawarosenbaum2018}.

Let $X_t=\log(S_t/S_0)$. In the risk-neutral, zero-rate normalization used
below, the continuous rough Heston model can be written as
\begin{equation}
\label{eq:intro-rough-heston}
\begin{aligned}
 &X_t
 = -\frac12\int_0^t V_s\dd s
    +\int_0^t\sqrt{V_s}
      \left(\sqrt{1-\rho^2}\dd W_{1,s}+\rho\dd W_{2,s}\right),\\
 &V_t
 =g_0(t)+b\int_0^tK_\alpha(t-s)V_s\dd s
   +\sqrt c\int_0^tK_\alpha(t-s)\sqrt{V_s}\dd W_{2,s},\\
 &K_\alpha(t)=\frac{t^{\alpha-1}}{\Gamma(\alpha)},
 \qquad \alpha\in(1/2,1),
\end{aligned}
\end{equation}
where $W_1$ and $W_2$ are independent Brownian motions, $g_0$ is the
deterministic initial variance curve, $b$ is the Volterra drift loading,
$\sqrt c$ is the volatility-of-volatility coefficient, and $\rho$ controls
return--variance leverage. Thus the model replaces the Markovian square-root
variance dynamics by an affine stochastic Volterra equation with power-law
memory. It remains analytically tractable: its Fourier--Laplace transform is
characterized by a fractional Riccati equation
\citep{eleuchrosenbaum2019,abijaberlarssonpulido2019}.
\citet{eleuchrosenbaum2018} obtain perfect hedges for European options by
trading the stock together with forward-variance exposures. Stock-only
hedging instead yields an orthogonal volatility-risk component. Related
perfect- and partial-hedging approaches
under rough volatility are discussed by \citet{fukasawahorvathtankov2021} and
\citet{motte2024partial}.

The rough Hawkes--Heston model introduced by \citet{bondi2024} adds
state-dependent common jumps to~\eqref{eq:intro-rough-heston}. It specializes
the affine Volterra jump framework of \citet{bondi2024affine} to a
price--variance system with endogenous jump clustering. In the notation of
this paper, its reduced-form dynamics are
\begin{equation}
\label{eq:intro-rhh}
\begin{aligned}
 \frac{\dd S_t}{S_{t-}}
 &=\sqrt{V_t}\left(\sqrt{1-\rho^2}\dd W_{1,t}
                   +\rho\dd W_{2,t}\right)
   +\int_{(0,\infty)}\left(e^{-\Lambda z}-1\right)
        \mut(\dd t,\dd z),\\
 V_t
 &=g_0(t)+b\int_0^tK(t-s)V_s\dd s
   +\sqrt c\int_0^tK(t-s)\sqrt{V_s}\dd W_{2,s}\\
 &\quad+\int_{[0,t)}\int_{(0,\infty)}K(t-s)z\,
        \mut(\dd s,\dd z),\\
 \mut(\dd t,\dd z)
 &:=\mu(\dd t,\dd z)-V_t\dd t\,\nu(\dd z).
\end{aligned}
\end{equation}
Here $\mu$ is a marked point measure, its intensity is proportional to spot
variance, $z>0$ is the jump mark, and $\Lambda\geq0$ controls the associated
downward log-price jump $-\Lambda z$. The same compensated random measure
therefore produces a positive innovation in the variance driver and a
downward stock-price jump, so large volatility moves cluster endogenously.
The fractional model corresponds to $K=K_\alpha$. When $\nu=0$, the jump
terms vanish and~\eqref{eq:intro-rhh} reduces to the continuous rough Heston
system~\eqref{eq:intro-rough-heston}.

The resulting model jointly fits SPX and VIX smiles while preserving
Fourier--Laplace tractability \citep{bondi2024}. Stock trading spans the jump exposure
$z\mapsto e^{-\Lambda z}-1$, whereas a derivative generally has a different
nonlinear exposure to the mark $z$. The quadratic hedge therefore selects the
best stock projection of the claim's Brownian and marked-jump risks.

The reduced-form model also admits a microscopic foundation.
\citet{wangwuzhu2026} derive a rough Hawkes--Heston limit from a
Poisson-embedded marked-Hawkes order-flow model. This provides an order-flow
interpretation of the limiting price and variance dynamics. Our hedging
analysis starts directly from the reduced-form asset-price model.

We study the optimal stock hedge in this given market and the construction
of approximations that attain its minimum quadratic error in the limit.
The central approximation question is how a regularized hedge, evaluated
using the original market's observed history, performs when traded against
the original stock. This question links the model's conditional-transform
formulas to the stability of the hedging decision under kernel
regularization.

Quadratic hedging is a natural response to this incompleteness. Given a
square-integrable payoff $Y$ and assuming that the discounted stock-price
process $S$ is a square-integrable $\Q$-martingale, the variance-optimal
stock-only strategy minimizes the following criterion. Here $\vartheta_t$
denotes the number of shares held immediately before time $t$; its
predictability rules out reacting to an unobserved jump. Starting from initial
capital $x$, the self-financing discounted terminal wealth is
$x+\int_0^T\vartheta_t\dd S_t$, where the stochastic integral is the cumulative
gain from trading in the stock. Thus the optimization problem is
\begin{equation}
\label{eq:intro-objective}
 \inf_{x\in\R,\,\vartheta\in L^2(S)}
 \E^{\Q}\left[
   \left(Y-x-\int_0^T\vartheta_t\dd S_t\right)^2
 \right].
\end{equation}
For general claims, this mean--variance criterion and its formulation as an
$L^2$ approximation by stochastic integrals are developed by
\citet{schweizer1992,schweizer1994}; the associated optimal initial capital
and variance-optimal martingale measure are analyzed by
\citet{schweizer1996}.
In the martingale case, this quadratic projection approach goes back to
\citet{foellmersondermann1986}; \citet{bouleaulamberton1989} subsequently
derived explicit residual-risk and optimal-hedge formulas in Markovian
markets. Writing
$M_t^Y:=\E^\Q[Y|\F_t]$, the classical Galtchouk--Kunita--Watanabe (GKW)
decomposition \citep{galtchouk1976,kunitawatanabe1967} gives the orthogonal
projection of the centered claim martingale
$M^Y-M_0^Y$ onto the closed stable space generated by $S$. Its
stochastic-integral component is the optimal strategy, while
$x^*=M_0^Y=\E^\Q[Y]$. Quadratic hedging
and its relation to local risk minimization and mean-variance hedging are
surveyed by \citet{schweizer2001}; \citet{heathplatenschweizer2001} compare
the two approaches in stochastic-volatility models, and general semimartingale
results are developed in \citet{cernykallsen2007}. In the present risk-neutral setting, the GKW
projection is especially transparent because the discounted stock has zero
drift. Under a historical measure, an additional opportunity process and a
variance-optimal martingale measure would be required. Throughout this paper,
variance optimality is therefore understood relative to the measure $\Q$ in~\eqref{eq:intro-objective}.

Our approach builds on the transform-based GKW methodology developed for
affine stochastic-volatility models by \citet{kallsenpauwels2010} and extended
to more general factor models and semi-static hedging by
\citet{ditellahauboldkellerressel2019}. In the rough Hawkes--Heston model,
the conditional transform of \citet{bondi2024} supplies the Brownian and
common marked-jump integrands needed for this projection. For each complex
exponent $w$, define
\[
 H_t(w,T):=\E^{\Q}[e^{wX_T}|\F_t].
\]
The process $H(w,T)$ is an exponential-affine martingale whose Brownian and
jump integrands are explicit functions of the Riccati--Volterra solution
$\psi_w$. Projecting these integrands on those of $S$ gives a closed formula
for the optimal stock position. The calculation is one-dimensional even
though the volatility state is infinite-dimensional. European claims are then
handled by Fourier inversion. In particular, the Lewis representation used by
\citet{bondi2024}, following \citet{lewis2001}, places the contour at
$\operatorname{Re}w=1/2$, exactly where
$e^{wX_T}$ is square-integrable whenever $S$ is a martingale.

There are two especially close model-specific strands of literature. First,
\citet{martins2023thesis} studies mean-variance hedging in the continuous rough
Heston model under the historical measure. Because the stock then has a drift,
the solution requires an opportunity process, a variance-optimal martingale
measure, and a feedback equation. Our martingale formulation under $\Q$
leads directly to a GKW projection. Second,
\citet{callegaro2025semistatic} develops Fourier and GKW methods for
semi-static variance-optimal hedging in a finite-dimensional affine model
with self-exciting jumps, with a variance swap hedged by dynamic stock
trading and static European options. Its dynamic inner problem provides a
close methodological precedent for our stock projection. Their specification
uses a continuous Heston variance and a separate Hawkes intensity. In the
rough Hawkes--Heston setting considered here, the conditional transform depends
on a Volterra history, and common marked jumps enter that history through
a potentially singular kernel. The approximation analysis therefore
requires a predictable reconstruction of this history and control of the
resulting trading error against the original stock.

Approximation results provide a separate point of comparison. Markovian
representations and finite-dimensional approximations of Volterra Heston and
fractional-kernel equations are developed by
\citet{abijabereleuch2019markovian,abijabereleuch2019,
abijaber2019lifting,bayerbreneis2023}; efficient weak
simulation of rough Heston is studied by \citet{bayerbreneis2024simulation}.
More specifically, \citet{bayerbreneis2023weak} bound weak pricing errors for
continuous rough Heston in terms of the $L^1$ kernel error, while
\citet{dinunnoyurchenko2026} derive approximation error estimates for
mean-square hedging in the continuous sandwiched Volterra volatility model.
Our stability result addresses the state-dependent common marked jumps of the
rough Hawkes--Heston model while keeping its stock, variance driver, initial
curve, and information flow fixed. The regularized kernel is used in the
Riccati equation and the reconstruction of the hedge's memory. Each
approximate strategy is thus evaluated in the same trading space $L^2(S)$,
with terminal error measured against the same payoff and stock.

Three features of this construction determine the stability proof. The
singularity of the kernel at a jump time is handled by histories integrated
strictly before the trading time and an explicit predictable version of the
memory process. A modulus truncation transfers the exact conditional
transform's random envelope to its approximation. Finally, truncating the
integrable weight $S_{t-}V_t$ converts the memory estimates into convergence
in the stock's trading norm. These steps connect deterministic kernel
regularization to the admissibility and performance of the resulting
stochastic trading strategy.

The paper makes two contributions.
\begin{enumerate}[label=(\roman*)]
\item \textit{Stock hedges and residual risk.}
      We specialize the affine-transform GKW construction to the rough
      Hawkes--Heston model, deriving the Brownian and common marked-jump
      covariance terms and the orthogonal residual for exponential claims.
      Fourier synthesis gives semi-explicit call hedges and a spectral
      representation of the minimum quadratic error; see
      Section~\ref{sec:main-results}.

\item \textit{Convergence of predictable approximations in the original market.}
      We construct admissible kernel-regularized holdings from the original
      history. Lemma~\ref{lem:original-deterministic-stability} controls the
      Riccati and memory terms, and
      Theorem~\ref{thm:original-compact-hedge} proves convergence in $L^2(S)$
      on compact Fourier intervals. For calls,
      Corollary~\ref{cor:original-call-stability} constructs a joint sequence
      of kernels and Fourier cutoffs along which initial capitals and
      strategies converge, continuous-time gains converge uniformly in time
      in square mean, and terminal mean-square errors converge to the
      original market's variance-optimal value.
\end{enumerate}

The numerical experiment illustrates the second contribution on common
original-market paths using shifted fractional kernels. It follows the
proof through kernel, Riccati, memory, strategy, and gain diagnostics and
measures the combined squared capital and strategy distance to a numerical
benchmark.

The rest of the paper is organized as follows. Section~\ref{sec:model} states
the rough Hawkes--Heston model and the quadratic hedging problem.
Section~\ref{sec:main-results} gives the exponential-claim projection, the
Fourier strategy, and the hedging-error formula.
Section~\ref{sec:original-market-stability} constructs approximate hedges
and proves their convergence in the original market.
Section~\ref{sec:original-market-numerics} gives a fixed-market numerical
experiment aligned with that construction.
Section~\ref{sec:conclusion} concludes.

\section{Model and quadratic hedging problem}
\label{sec:model}

\subsection{The rough Hawkes--Heston model}

We use the risk-neutral specification of \citet{bondi2024} with zero interest
rate and dividends. Let $(\Omega,\F,\mathbb F,\Q)$ satisfy the usual
conditions, let $\F_0$ be $\Q$-trivial, and suppose that the space supports two
independent Brownian motions $W_1,W_2$ and an
integer-valued random measure $\mu(\dd t,\dd z)$ on
$\R_+\times(0,\infty)$. The following assumption reproduces the kernel and
initial-curve conditions in \citet[Hypotheses~2.1--2.2]{bondi2024}, together
with the standing jump and variance-process conditions in their model setup.

\begin{assumption}[Model admissibility]
\label{ass:model}
The kernel $K\in L^2_{\mathrm{loc}}(\R_+)$ is nonnegative, nonincreasing,
not identically zero, and continuously differentiable on $(0,\infty)$. It has
a resolvent of the first kind $L$ which is nonnegative and nonincreasing in
the sense that $s\mapsto L([s,s+h])$ is nonincreasing for every $h\geq0$.
The deterministic curve $g_0$ is continuous and nondecreasing with
$g_0(0)\geq0$, and $\nu$ is a nonnegative Borel measure on $(0,\infty)$ such
that
\[
 \int_{(0,\infty)}z^2\nu(\dd z)<\infty.
\]
Finally, a nonnegative predictable process $V$ with paths in
$L^2_{\mathrm{loc}}(\R_+)$ is given which solves the stochastic Volterra
equation below in the weak sense.
\end{assumption}

Write $V_t=\sigma_t^2$ for spot variance. The predictable compensator of
$\mu$ is
\begin{equation}
\label{eq:compensator}
 \nu^\mu(\dd t,\dd z)=V_t\dd t\,\nu(\dd z),
 \qquad
 \mut(\dd t,\dd z)=\mu(\dd t,\dd z)-V_t\dd t\,\nu(\dd z).
\end{equation}
The variance process solves the affine Volterra equation
\begin{align}
 V_t&=g_0(t)+(K*\dd Z)_t,
 \label{eq:variance-volterra}\\
 \dd Z_t&=bV_t\dd t+\sqrt{c}\sqrt{V_t}\dd W_{2,t}
       +\int_{(0,\infty)}z\mut(\dd t,\dd z).
 \label{eq:variance-driver}
\end{align}
Here $b\in\R$ and $c>0$. Equation~\eqref{eq:variance-volterra} is understood
$\Q\otimes\dd t$-almost everywhere. We use the predictable stochastic
convolution convention
\[
 (K*\dd Z)_t:=\int_{[0,t)}K(t-s)\dd Z_s;
\]
this convention is important for a singular kernel at a jump time. Changing
the value at the right endpoint does not alter the
$\Q\otimes\dd t$-equivalence class.

\begin{assumption}[Conditional-transform regularity]
\label{ass:transform}
The first-kind resolvent has the decomposition
\[
 L(\dd t)=\ell(t)\dd t+\ell_0\delta_0(\dd t),
 \qquad \ell\in L^1_{\mathrm{loc}}(\R_+),\quad \ell_0\geq0.
\]
Moreover, writing $\Delta_\varepsilon K(t)=K(t+\varepsilon)$,
\[
 \sup_{\varepsilon\in(0,T]}
 \|\Delta_\varepsilon K*L\|_{\mathrm{TV}([0,T])}<\infty
 \qquad\text{for every }T>0.
\]
\end{assumption}

The kernel clause of Assumption~\ref{ass:model} is precisely
\citet[Hypothesis~2.1]{bondi2024}. Its monotonicity, differentiability, and
resolvent requirements are also used in the existence and affine-transform
theory of \citet{abijaberlarssonpulido2019,bondi2024affine}. Its
$g_0$ clause is \citet[Hypothesis~2.2]{bondi2024}. Writing $\nu$ on
$(0,\infty)$ is equivalent to their convention of taking a measure on
$\R_+$ with $\nu(\{0\})=0$. Assumption~\ref{ass:transform} consists exactly
of the additional kernel conditions imposed in
\citet[Theorem~3.3]{bondi2024}. It supplies the regularity used by the
conditional transform, while Assumption~\ref{ass:model} contains the kernel
conditions for the model and the Riccati--Volterra equation.

We assume weak existence directly in Assumption~\ref{ass:model}. If existence
is instead to be deduced from primitive conditions, then
\citet[Remark~2.3]{bondi2024} additionally requires $K(\cdot+1/n)$ to satisfy
their Hypothesis~2.1 for every $n\in\mathbb N$; weak existence then follows
from \citet[Theorem~2.13]{abijaber2021weak} and
\citet[Lemma~9]{bondi2024affine}. Once existence is given, weak uniqueness
under the kernel conditions in Assumption~\ref{ass:model} follows from
\citet[Corollary~12]{bondi2024affine}. The results below therefore start from
the weak solution specified in Assumption~\ref{ass:model}.

The moment estimate of \citet[Lemma~1]{bondi2024affine}, applied to the
scalar affine characteristics $b(v)=bv$, $a(v)=cv$, and
$\eta(v,\dd z)=v\nu(\dd z)$, also gives
\begin{equation}
\label{eq:original-integrated-first-moment}
 \E^\Q\!\left[\int_0^T V_t\dd t\right]<\infty,
 \qquad T<\infty.
\end{equation}
Its hypotheses follow from the continuous initial curve, the locally
square-integrable kernel, and the jump second moment in
Assumption~\ref{ass:model}. This bound will control the approximation of
the hedge using the original variance history.

The fractional specification is
\begin{equation}
\label{eq:fractional-kernel}
 K(t)=K_\alpha(t):=\frac{t^{\alpha-1}}{\Gamma(\alpha)},
 \qquad \alpha\in(1/2,1].
\end{equation}
Here $\alpha>1/2$ is exactly the condition $K_\alpha\in
L^2_{\mathrm{loc}}(\R_+)$. For $\alpha\in(1/2,1)$, the first-kind resolvent is
\[
 L_\alpha(\dd t)=\frac{t^{-\alpha}}{\Gamma(1-\alpha)}\dd t,
\]
whereas $L_1=\delta_0$. Hence the fractional specification satisfies all the
kernel clauses of Assumptions~\ref{ass:model} and~\ref{ass:transform}: for
$\alpha<1$, $\Delta_\varepsilon K_\alpha*L_\alpha$ is nondecreasing and
bounded above by $K_\alpha*L_\alpha=1$, and the case $\alpha=1$ is immediate.
The boundary case $\alpha=1$ has $K_1\equiv1$ and is Markovian. When
$g_0(t)=V_0+\beta t$, the corresponding c\`adl\`ag state modification
satisfies a jump-CIR equation. This is a useful benchmark, although it is not
identical to the Heston--Hawkes model of \citet{banos2024}, which has a
separate exponential Hawkes intensity state.

Let $X_t=\log(S_t/S_0)$. The stock dynamics can be written directly in
martingale form as
\begin{equation}
\label{eq:stock-martingale}
 \frac{\dd S_t}{S_{t-}}
 =\sqrt{V_t}\left(\sqrt{1-\rho^2}\dd W_{1,t}+\rho\dd W_{2,t}\right)
  +\int_{(0,\infty)}q(z)\mut(\dd t,\dd z),
 \qquad q(z):=e^{-\Lambda z}-1,
\end{equation}
where $\rho\in[-1,1]$ and $\Lambda\ge0$. Thus a mark $z$ creates the common
log-price jump $\Delta X_t=-\Lambda z$ and a positive innovation in the
variance driver. Since $q(z)>-1$, the Dol\'eans exponential representation
of~\eqref{eq:stock-martingale} also shows that $S_t>0$ on finite horizons.
When $\nu=0$, the jump terms vanish and the model reduces to continuous rough
Heston.
Equivalently,
\begin{align}
 \dd X_t
 &= -\left[\frac12+\int_{(0,\infty)}
       \left(e^{-\Lambda z}-1+\Lambda z\right)\nu(\dd z)\right]V_t\dd t
 \nonumber\\*
 &\quad+\sqrt{V_t}\left(\sqrt{1-\rho^2}\dd W_{1,t}
                   +\rho\dd W_{2,t}\right)
       -\Lambda\int_{(0,\infty)}z\mut(\dd t,\dd z).
 \label{eq:log-price}
\end{align}

For $u,v\in\C$ in their admissible domain, define the affine characteristic
\begin{align}
\Rr(u,v)
&:=\frac12(u^2-u)+(b+\rho\sqrt c\,u)v+\frac c2v^2
\nonumber\\
&\quad+\int_{(0,\infty)}
 \left(e^{(v-\Lambda u)z}-u(e^{-\Lambda z}-1)-1-vz\right)\nu(\dd z).
\label{eq:R-map}
\end{align}
For $w\in\C$ with $\operatorname{Re}w\in[0,1]$, the associated
Riccati--Volterra equation is
\begin{equation}
\label{eq:riccati-volterra}
 \psi_w(t)=\int_0^tK(t-s)\Rr(w,\psi_w(s))\dd s.
\end{equation}
By \citet[Theorem~3.1]{bondi2024}, the kernel conditions in
Assumption~\ref{ass:model} ensure that this equation has a unique continuous
solution with $\operatorname{Re}\psi_w\leq0$. Define the adjusted forward process
\begin{equation}
\label{eq:adjusted-forward}
 g_t(s)=g_0(s)+\int_{[0,t]}K(s-r)\dd Z_r,
 \qquad s>t.
\end{equation}
Under Assumptions~\ref{ass:model} and~\ref{ass:transform},
\citet[Theorem~3.3]{bondi2024} gives the following c\`adl\`ag version of the
conditional Fourier--Laplace transform, up to indistinguishability:
\begin{equation}
\label{eq:H-affine}
H_t(w,T):=\E^\Q[e^{wX_T}|\F_t]
=\exp\left(
  wX_t+\int_t^T\Rr(w,\psi_w(T-s))g_t(s)\dd s
\right).
\end{equation}

\subsection{The Hilbert-space hedging problem}

Let $S$ be the only dynamically traded risky asset. On the predictable
$\sigma$-field of $\Omega\times[0,T]$, define the measure
\[
 \mu_S(A):=\E^\Q\!\left[\int_0^T
   \mathbf 1_A(\omega,t)\dd\langle S\rangle_t\right].
\]
Denote by $L^2(S)$ the space of real-valued predictable strategies, identified up to
$\mu_S$-null sets, satisfying
\[
 \E^\Q\left[\int_0^T\vartheta_t^2\dd\langle S\rangle_t\right]<\infty.
\]
For complex claims we use the complexifications $L^2(\Q;\C)$ and
$L^2(S;\C)$, with
\[
 \|\vartheta\|_{L^2(S;\C)}^2
 :=\E^\Q\left[\int_0^T|\vartheta_t|^2\dd\langle S\rangle_t\right],
 \qquad
 \langle\vartheta,\eta\rangle_{L^2(S;\C)}
 :=\E^\Q\left[\int_0^T\vartheta_t\overline{\eta_t}\dd\langle S\rangle_t\right].
\]
We suppress $\C$ in the space notation when the context is clear. Complex
GKW decompositions are obtained by applying the real decomposition to the
real and imaginary parts separately.

\begin{assumption}[Square-integrability]
\label{ass:l2}
The stock $S$ is a square-integrable $\Q$-martingale on $[0,T]$,
and every target payoff $Y$ under consideration belongs to $L^2(\Q)$.
\end{assumption}

The $L^2$ projection requires the stock second moment in addition to the
true-martingale property proved by
\citet[Corollary~3.4]{bondi2024}. This moment is governed by a
Riccati--Volterra equation and may explode even in continuous rough Heston
\citep{gerhold2019moment}. The following proposition gives a directly
verifiable criterion in terms of the model coefficients.
Assumption~\ref{ass:model} already implies
$\int_{(0,\infty)}q(z)^2\nu(\dd z)<\infty$.

\begin{proposition}[A model-based second-moment criterion]
\label{prop:l2-sufficient}
Suppose that Assumptions~\ref{ass:model} and~\ref{ass:transform} hold, and set
\begin{align}
 d_2&:=1+\int_{(0,\infty)}q(z)^2\nu(\dd z),
 \nonumber\\
 \widehat b_2&:=b+2\rho\sqrt c
   +\int_{(0,\infty)}z\left(e^{-2\Lambda z}-1\right)\nu(\dd z),
 \label{eq:d2-b2}\\
 F_2(v)&:=d_2+\widehat b_2v+\frac c2v^2
   +\int_{(0,\infty)}e^{-2\Lambda z}
      \left(e^{vz}-1-vz\right)\nu(\dd z).
 \label{eq:F2}
\end{align}
Thus $F_2(v)=\Rr(2,v)$. Consider the Riccati--Volterra equation
\begin{equation}
 \psi_2(t)=\int_0^t K(t-s)F_2\left(\psi_2(s)\right)\dd s,
 \label{eq:psi2}
\end{equation}
and suppose that there exists $\bar v>0$ such that
\begin{equation}
 F_2(\bar v)\le0,
 \qquad
 \text{and either }\nu=0\text{ or }\bar v<2\Lambda.
 \label{eq:F2-barrier}
\end{equation}
Then~\eqref{eq:psi2} has a unique global continuous solution satisfying
\begin{equation}
 0\le\psi_2(t)\le v_*\le\bar v,
 \qquad t\ge0,
 \label{eq:psi2-bounds}
\end{equation}
where $v_*$ is the first positive zero of $F_2$. Consequently, $S$ is a
square-integrable $\Q$-martingale on every finite horizon and, for every
$T>0$,
\begin{equation}
 \E[S_T^2]
 =S_0^2\exp\left(
   \int_0^T F_2\left(\psi_2(T-s)\right)g_0(s)\dd s
 \right)<\infty.
 \label{eq:stock-second-moment}
\end{equation}
\end{proposition}

\begin{proof}
The identity $F_2(v)=\Rr(2,v)$ follows by collecting the constant and linear
terms in~\eqref{eq:R-map}. Since $F_2(0)=d_2>0$, continuity
and~\eqref{eq:F2-barrier} give a first positive zero $v_*\le\bar v$. Moreover,
$F_2$ is locally Lipschitz on a neighborhood of $[0,\bar v]$: when
$\nu\ne0$, this follows from $\bar v<2\Lambda$ and
$\int_{(0,\infty)}z^2\nu(\dd z)<\infty$, while for $\nu=0$ the function is quadratic.
The local-existence theorem of
\citet[Chapter~12, Theorem~1.1]{gripenberg1990} therefore gives a
noncontinuable continuous solution of~\eqref{eq:psi2}.

To see that this solution remains in $[0,v_*]$, write, along the solution,
\[
 F_2(\psi_2)=d_2+a_0\psi_2,
 \qquad
 -F_2(\psi_2)=a_1(v_*-\psi_2),
\]
where the difference quotients $a_0$ and $a_1$ are extended continuously at
$0$ and $v_*$, respectively. On every compact time interval before the
maximal existence time these coefficients are bounded. The positivity result
for linear Volterra equations in
\citet[Theorem~C.1 and Remark~B.6]{abijabereleuch2019}, first applied to
$\psi_2=K*(a_0\psi_2+d_2)$ and then to
$v_*-\psi_2=v_*+K*(a_1(v_*-\psi_2))$, yields
$0\le\psi_2\le v_*$. Here the constant forcing curve $g\equiv v_*$ is
admissible under the kernel conditions in Assumption~\ref{ass:model}. Since
the solution remains in a compact subset of the domain of $F_2$, the
continuation argument makes it global; local Lipschitz continuity also gives
uniqueness. This is the same local-existence, comparison, and continuation
method used in the proof of \citet[Theorem~3.1]{bondi2024}.

Now fix $T>0$ and define the real semimartingale
\begin{align*}
 U_0&:=\int_0^T F_2\left(\psi_2(T-s)\right)g_0(s)\dd s,\\
 U_t&:=U_0+2X_t+\int_0^t\psi_2(T-s)\dd Z_s
       -\int_0^tF_2\left(\psi_2(T-s)\right)V_s\dd s.
\end{align*}
 After localization, an application of the stochastic Fubini theorem (see, e.g.,
 \citet[Theorem~IV.65]{protter2005}) to~\eqref{eq:variance-volterra},
 together with the Riccati equation~\eqref{eq:psi2}, gives
\[
 \int_0^T F_2\left(\psi_2(T-s)\right)(V_s-g_0(s))\dd s
 =\int_0^T\psi_2(T-s)\dd Z_s.
\]
Hence $U_T=2X_T$. It\^o's formula shows that $\widehat H_t:=e^{U_t}$
has stochastic logarithm
\begin{align*}
 \frac{\dd\widehat H_t}{\widehat H_{t-}}
 ={}&2\sqrt{1-\rho^2}\sqrt{V_t}\,\dd W_t^1
   +\left(2\rho+\sqrt c\,\psi_2(T-t)\right)
      \sqrt{V_t}\,\dd W_t^2 \\*
 &+\int_{(0,\infty)}
   \left(e^{(\psi_2(T-t)-2\Lambda)z}-1\right)
   \widetilde\mu(\dd t,\dd z).
\end{align*}
The deterministic coefficients are bounded, and the jump coefficient is
nonpositive when $\nu\ne0$ by~\eqref{eq:psi2-bounds}
and~\eqref{eq:F2-barrier}. Therefore the stochastic-exponential criterion
of \citet[Lemma~3.2]{bondi2024} makes $\widehat H$ a true martingale (with the
jump term absent when $\nu=0$). Since $\widehat H_T=e^{2X_T}$, evaluating its
expectation at time zero yields~\eqref{eq:stock-second-moment}.
\end{proof}

\begin{remark}[Verification for the calibrated specification]
\label{rem:calibrated-l2}
For the calibration reported by \citet{bondi2024},
\[
 \nu(\dd z)=e^{-z}\dd z,
 \qquad
 (b,\rho,c,\Lambda)=(-1.812,-0.731,0.115,0.276).
\]
Writing $a=1+2\Lambda$, equations~\eqref{eq:d2-b2}--\eqref{eq:F2}
reduce to
\begin{align*}
 d_2&=2+\frac{1}{1+2\Lambda}-\frac{2}{1+\Lambda},\\
 \widehat b_2&=b+2\rho\sqrt c+\frac{1}{(1+2\Lambda)^2}-1,\\
 F_2(v)&=d_2+\widehat b_2v+\frac c2v^2
   +\frac{1}{a-v}-\frac{1}{a}-\frac{v}{a^2},
 \qquad v<a.
\end{align*}
Numerically,
\[
 d_2\approx1.0769,
 \qquad
 \widehat b_2\approx-2.8926,
 \qquad
 F_2(0.4)\approx-0.0133<0,
 \qquad
 0.4<2\Lambda=0.552.
\]
Thus condition~\eqref{eq:F2-barrier} holds with $\bar v=0.4$, and
Proposition~\ref{prop:l2-sufficient} verifies that the calibrated stock is a
square-integrable $\Q$-martingale on every finite horizon.
\end{remark}

Proposition~\ref{prop:l2-sufficient} gives a coefficient-level criterion for
Assumption~\ref{ass:l2}: the barrier controls the Riccati--Volterra solution
globally and thereby verifies the required stock second moment.

For $Y\in L^2(\Q)$, let $M_t^Y=\E^\Q[Y|\F_t]$. The GKW decomposition
\citep{galtchouk1976,kunitawatanabe1967} is
\begin{equation}
\label{eq:gkw-general}
 M_t^Y=M_0^Y+\int_0^t\vartheta_s^Y\dd S_s+L_t^Y,
 \qquad \langle L^Y,S\rangle=0,
\end{equation}
where $L^Y$ is a square-integrable martingale strongly orthogonal to $S$. It
solves~\eqref{eq:intro-objective} with
\begin{equation}
\label{eq:optimal-capital-error}
 x^*=\E^\Q[Y],\qquad
 \vartheta^*=\vartheta^Y,
 \qquad
 \varepsilon_Y^2
 :=\E^\Q\left[
   \left(Y-x^*-\int_0^T\vartheta_t^*\dd S_t\right)^2
 \right]
 =\E^\Q[(L_T^Y)^2].
\end{equation}
Indeed, constants are orthogonal in $L^2(\Q)$ to terminal stochastic
integrals against the martingale $S$, and the GKW decomposition is precisely
the orthogonal decomposition of $Y-\E^\Q[Y]$ into the closed subspace of such
integrals and its orthogonal complement. Equation~\eqref{eq:optimal-capital-error}
then follows from the Pythagorean identity.

\paragraph{Hedging framework.}
All optimality statements refer to the continuous-time gain space
\[
 \mathcal G^2(S)=
 \left\{\int_0^T\vartheta_t\dd S_t:\vartheta\in L^2(S)\right\}
\]
and the $L^2(\Q)$ criterion in~\eqref{eq:intro-objective}. Thus the results characterize
continuous trading under the risk-neutral measure; discrete-time and
historical-measure hedging are separate optimization problems.

\section{Stock hedges and residual risk}
\label{sec:main-results}

We apply the affine conditional transform of \citet{bondi2024} within the
transform-based GKW framework to obtain the stock hedge and its orthogonal
residual. The explicit covariance terms account for both Brownian shocks
and common marked jumps. These formulas also specify the target strategies
for the approximation analysis in
Section~\ref{sec:original-market-stability}.

\subsection{Projection of the affine Fourier martingales}

For complex local martingales we use the bilinear extension of predictable
covariation. Thus $\langle M,\overline N\rangle$ is the Hermitian bracket and
$\E[\langle M,\overline M\rangle_T]\geq0$ whenever $M_0=0$ and $M_T$ is
square-integrable.

Fix $T>0$ and $w\in\C$ with $\operatorname{Re}w\in[0,1]$ such that
$e^{wX_T}\in L^2(\Q)$. Then $H(w,T)$ in~\eqref{eq:H-affine} is a
square-integrable complex martingale. To simplify notation, set
\begin{align}
 \psi_t^w&:=\psi_w(T-t),
 \label{eq:psi-short}\\
 r_t^w(z)&:=\exp\left((\psi_t^w-\Lambda w)z\right)-1,
 \label{eq:r-short}\\
 D&:=1+\int_{(0,\infty)}q(z)^2\nu(\dd z),
 \label{eq:D-def}\\
 B_t^w&:=w+\rho\sqrt c\,\psi_t^w
       +\int_{(0,\infty)}r_t^w(z)q(z)\nu(\dd z),
 \label{eq:B-def}\\
 \beta_t^w&:=\frac{B_t^w}{D}.
 \label{eq:beta-def}
\end{align}
Because $\operatorname{Re}\psi_t^w\leq0$ and
$\operatorname{Re}w\geq0$, both $q$ and $r_t^w$ belong to $L^2(\nu)$:
near zero their absolute values are bounded by a constant times $z$, while
away from zero they are bounded and $\nu([1,\infty))<\infty$. Thus every
quantity in~\eqref{eq:D-def}--\eqref{eq:beta-def} is finite, with $D\geq1$.
The first two terms in $B_t^w$ are the instantaneous Brownian covariance
between $H(w,T)$ and the stock. The integral is the covariance generated by
their common marked jumps. The denominator $D$ is the normalized predictable
quadratic variation of the stock.

\begin{theorem}[Variance-optimal hedge for an exponential claim]
\label{thm:exponential-gkw}
Suppose Assumptions~\ref{ass:model}, \ref{ass:transform}, and~\ref{ass:l2}
hold and let $w$ satisfy the conditions above. Then the complex GKW
decomposition of $e^{wX_T}$ with respect to $S$ is obtained from the
conditional-transform martingale recalled from~\eqref{eq:H-affine}:
\[
 H_t(w,T)=\E^\Q[e^{wX_T}|\F_t]
 =\exp\left(
   wX_t+\int_t^T\Rr(w,\psi_w(T-s))g_t(s)\dd s
 \right).
\]
\begin{samepage}
Specifically,
\begin{equation}
\label{eq:gkw-exponential}
 H_t(w,T)=H_0(w,T)+\int_0^t\vartheta_s^w\dd S_s+L_t^w,
\end{equation}
where
\begin{equation}
 \vartheta_t^w=
 \frac{H_{t-}(w,T)}{S_{t-}}
 \frac{
   w+\rho\sqrt c\,\psi_w(T-t)
   +\displaystyle\int_{(0,\infty)}
      \left(e^{(\psi_w(T-t)-\Lambda w)z}-1\right)
      \left(e^{-\Lambda z}-1\right)\nu(\dd z)
 }{
   1+\displaystyle\int_{(0,\infty)}
      \left(e^{-\Lambda z}-1\right)^2\nu(\dd z)
 }.
\label{eq:theta-explicit}
\end{equation}
The residual $L^w$ is a square-integrable complex martingale strongly
orthogonal to $S$.
\end{samepage}
\end{theorem}

\begin{proof}
The martingale calculation in \citet[Theorem~3.3]{bondi2024} gives
\begin{align}
 \frac{\dd H_t(w,T)}{H_{t-}(w,T)}
 &=\sqrt{V_t}\left[
    w\sqrt{1-\rho^2}\dd W_{1,t}
   +(w\rho+\sqrt c\,\psi_t^w)\dd W_{2,t}
   \right]
\nonumber\\
&\quad+\int_{(0,\infty)}r_t^w(z)\mut(\dd t,\dd z).
\label{eq:H-dynamics-proof}
\end{align}
Combining~\eqref{eq:H-dynamics-proof} with~\eqref{eq:stock-martingale}, and
using the compensator~\eqref{eq:compensator}, yields the following predictable
covariations:
\begin{align}
 \dd\langle H(w,T),S\rangle_t
 &=H_{t-}(w,T)S_{t-}V_tB_t^w\dd t,
 \label{eq:cross-bracket-proof}\\
 \dd\langle S\rangle_t
 &=S_{t-}^2V_tD\dd t.
 \label{eq:S-bracket-proof}
\end{align}
They are well defined under Assumptions~\ref{ass:model},
\ref{ass:transform}, and~\ref{ass:l2}, and the square-integrability of
$e^{wX_T}$. The Kunita--Watanabe inequality implies
that $\langle H(w,T),S\rangle$ is absolutely continuous with respect to
$\langle S\rangle$ in the sense required for the GKW projection. Therefore,
$\mu_S$-almost everywhere, its Radon--Nikodym derivative is
\[
 \vartheta_t^w
 =\frac{\dd\langle H(w,T),S\rangle_t}
        {\dd\langle S\rangle_t}
 =\frac{H_{t-}(w,T)}{S_{t-}}\frac{B_t^w}{D},
\]
which proves~\eqref{eq:theta-explicit}. On $\{V_t=0\}$ both bracket densities
vanish, and~\eqref{eq:theta-explicit} fixes an arbitrary predictable version.
Since $H(w,T)$ and $S$ are square-integrable, the GKW projection theorem
shows that this bracket density belongs to $L^2(S)$. Consequently,
$L^w:=H(w,T)-H_0(w,T)-\int_0^\cdot\vartheta_t^w\dd S_t$ is
square-integrable. Subtracting $\vartheta_t^w\dd S_t$
from~\eqref{eq:H-dynamics-proof} gives
\begin{align}
\dd L_t^w
&=H_{t-}(w,T)\sqrt{V_t}
  \left[
   (w-\beta_t^w)\sqrt{1-\rho^2}\dd W_{1,t}
   +(w\rho+\sqrt c\,\psi_t^w-\rho\beta_t^w)\dd W_{2,t}
  \right]
\nonumber\\
&\quad+H_{t-}(w,T)\int_{(0,\infty)}
       \left(r_t^w(z)-\beta_t^wq(z)\right)\mut(\dd t,\dd z).
\label{eq:residual-dynamics}
\end{align}
Finally,
\[
 \dd\langle L^w,S\rangle_t
 =H_{t-}(w,T)S_{t-}V_t(B_t^w-\beta_t^wD)\dd t=0.
\]
Thus the bracket identity above gives strong orthogonality, and uniqueness of
the GKW decomposition completes the proof.
\end{proof}

%\begin{remark}[Complexification]
%The quadratic problem in~\eqref{eq:intro-objective} is stated for real-valued claims and strategies. Theorem~\ref{thm:exponential-gkw} is an intermediate complexified identity, obtained by applying the real GKW projection separately to the real and imaginary parts. On the Lewis contour $w_\lambda=1/2+\mathrm{i}\lambda$, conjugate symmetry pairs the contributions at $\lambda$ and $-\lambda$, so the call strategy in Corollary~\ref{cor:call-hedge} is real-valued.
%\end{remark}

\begin{remark}[Interpretation of the jump correction]
A mark $z$ changes the stock by the relative amount $e^{-\Lambda z}-1$ and
changes $H(w,T)$ by the relative amount
$e^{(\psi_w(T-t)-\Lambda w)z}-1$. Since the two jumps are driven by the same
random measure, the numerator term
\[
 \int_{(0,\infty)}
   \left(e^{(\psi_w(T-t)-\Lambda w)z}-1\right)
   \left(e^{-\Lambda z}-1\right)\nu(\dd z)
\]
is their contribution to the predictable covariance between $H(w,T)$ and the
stock. The denominator term
\[
 \int_{(0,\infty)}
   \left(e^{-\Lambda z}-1\right)^2\nu(\dd z)
\]
is the jump contribution to the stock's predictable quadratic variation.
Together these terms adjust the bracket ratio defining the GKW holding for the
claim's exposure to stock jumps. The adjustment vanishes when $\nu=0$.
\end{remark}

\paragraph{Residual risk dynamics.}
Equation~\eqref{eq:residual-dynamics} makes the unhedgeable component
in~\eqref{eq:gkw-exponential} explicit.
The Brownian and marked-jump coefficients are precisely the exposures left
after projecting those of $H(w,T)$ onto the stock exposures. Their predictable
covariation with $S$ vanishes because $B_t^w-\beta_t^wD=0$.

Define
\begin{align}
 A_t^w
 &:=(1-\rho^2)|w|^2+|\rho w+\sqrt c\,\psi_t^w|^2
    +\int_{(0,\infty)}|r_t^w(z)|^2\nu(\dd z),
 \label{eq:A-def}\\
 \Gamma_t^w&:=A_t^w-\frac{|B_t^w|^2}{D}.
 \label{eq:Gamma-def}
\end{align}
By the Cauchy--Schwarz inequality, $\Gamma_t^w\ge0$.

The preceding residual dynamics yield a closed form for the minimum error.

\begin{corollary}[Minimum error for an exponential claim]
\label{cor:error-exponential}
Under the assumptions of Theorem~\ref{thm:exponential-gkw},
\begin{equation}
\label{eq:error-exponential}
 \E^\Q[|L_T^w|^2]
 =\E^\Q\left[
   \int_0^T|H_{t-}(w,T)|^2V_t\Gamma_t^w\dd t
 \right].
\end{equation}
Thus $\Gamma_t^w$ is the instantaneous unspanned-risk density per unit of
$|H_{t-}|^2V_t$.
\end{corollary}

\begin{proof}
The predictable quadratic variation of~\eqref{eq:residual-dynamics} is
\[
 \dd\langle L^w,\overline{L^w}\rangle_t
 =|H_{t-}(w,T)|^2V_t
 \left(A_t^w-\frac{|B_t^w|^2}{D}\right)\dd t.
\]
Taking expectations and applying the martingale isometry
gives~\eqref{eq:error-exponential}.
\end{proof}

\begin{remark}[Sanity checks]
For $w=1$, the Riccati solution is $\psi_1=0$ and $r_t^1=q$. Consequently,
$B_t^1=D$, $\beta_t^1=1$, and the residual
in~\eqref{eq:residual-dynamics} vanishes. Since $H_t(1,T)=S_t/S_0$,
formula~\eqref{eq:theta-explicit} gives the exact hedge $1/S_0$, as it must for
the normalized stock claim $e^{X_T}=S_T/S_0$.

If $\nu=0$, all jump terms disappear and~\eqref{eq:Gamma-def} reduces to
\[
 \Gamma_t^w=c(1-\rho^2)|\psi_t^w|^2.
\]
Thus, for a nontrivial claim with $\psi_t^w\ne0$ and $|\rho|<1$, a stock-only
strategy generally leaves volatility Brownian risk. Perfect replication in
continuous rough Heston is obtained by adding forward-variance exposures
\citep{eleuchrosenbaum2018}.
\end{remark}

\subsection{Fourier-representable payoffs}

Let $I\subset\R$, let $a:I\to\C$ and $w:I\to\C$ be measurable, and suppose
that for almost every $\lambda\in I$
$\operatorname{Re}w(\lambda)\in[0,1]$,
$e^{w(\lambda)X_T}\in L^2(\Q)$, and that the $L^2(\Q)$-valued map
$\lambda\mapsto a(\lambda)e^{w(\lambda)X_T}$ is strongly measurable. Define
the square-integrable-martingale norm by
\begin{align}
 \|H(w,T)\|_{\mathcal H^2}^2
 &:=
 |H_0(w,T)|^2+
 \E^\Q\!\left[\langle H(w,T),\overline{H(w,T)}\rangle_T\right]
 \nonumber\\
 &=\E^\Q[|e^{wX_T}|^2].
 \label{eq:H2-norm}
\end{align}
Assume that the preceding $L^2(\Q)$-valued map is Bochner integrable, or
equivalently that
\begin{equation}
\label{eq:fourier-stochastic-fubini}
 \int_I|a(\lambda)|
 \|H(w(\lambda),T)\|_{\mathcal H^2}\dd\lambda<\infty,
\end{equation}
and define the payoff by the $L^2(\Q)$ Bochner integral
\begin{equation}
\label{eq:general-fourier-payoff}
 Y=\int_I a(\lambda)e^{w(\lambda)X_T}\dd\lambda.
\end{equation}
Indeed, by~\eqref{eq:H2-norm},
\begin{equation}
 \left\|a(\lambda)e^{w(\lambda)X_T}\right\|_{L^2(\Q)}
 =|a(\lambda)|\,\|H(w(\lambda),T)\|_{\mathcal H^2},
 \label{eq:fourier-bochner-norm}
\end{equation}
so~\eqref{eq:fourier-stochastic-fubini} guarantees that $Y$ is well defined
and that
\[
 \|Y\|_{L^2(\Q)}
 \leq \int_I |a(\lambda)|
 \|H(w(\lambda),T)\|_{\mathcal H^2}\dd\lambda<\infty.
\]
Condition~\eqref{eq:fourier-stochastic-fubini} is a Bochner-space formulation
of the sufficient conditions used to interchange a Fourier representation and
the GKW decomposition in
\citet[Theorem~4.6 and Remark~4.7]{ditellahauboldkellerressel2019}.

The integrability condition allows the modewise GKW decompositions to be
synthesized under the Fourier integral.

\begin{theorem}[Fourier synthesis of the optimal hedge]
\label{thm:fourier-hedge}
Suppose Assumptions~\ref{ass:model}, \ref{ass:transform}, and~\ref{ass:l2}, as
well as the conditions above, hold. Then
\begin{align}
 M_t^Y&=\int_Ia(\lambda)H_t(w(\lambda),T)\dd\lambda,
 \label{eq:MY-fourier}\\
 \vartheta_t^Y&=\int_Ia(\lambda)\vartheta_t^{w(\lambda)}\dd\lambda,
 \label{eq:theta-fourier}\\
 L_t^Y&=\int_Ia(\lambda)L_t^{w(\lambda)}\dd\lambda.
 \label{eq:L-fourier}
\end{align}
In particular,~\eqref{eq:theta-fourier} is a semi-explicit strategy: each
integrand requires only the conditional transform~\eqref{eq:H-affine} and the
deterministic Riccati--Volterra solution~\eqref{eq:riccati-volterra}.
\end{theorem}

\begin{proof}
Let $\mathsf J:L^2(\Q)\to\mathcal H^2$ map a terminal random variable $\xi$
to the martingale $(\E^\Q[\xi|\F_t])_{t\leq T}$. By~\eqref{eq:H2-norm},
$\mathsf J$ is a linear isometry. Since bounded linear
maps commute with Bochner integration,
\[
 \mathsf JY
 =\int_Ia(\lambda)\mathsf J(e^{w(\lambda)X_T})\dd\lambda
 =\int_Ia(\lambda)H(w(\lambda),T)\dd\lambda,
\]
which proves~\eqref{eq:MY-fourier}.

The maps taking a martingale to the stochastic-integral component and to the
zero-initial residual in its GKW decomposition are bounded linear orthogonal
projections in $\mathcal H^2$. In particular, the martingale isometry and the
contraction property give
\[
 \|\vartheta^{w}\|_{L^2(S)}
 =\left\|\int_0^\cdot\vartheta_s^{w}\dd S_s\right\|_{\mathcal H^2}
 \leq\|H(w,T)\|_{\mathcal H^2},
 \qquad
 \|L^{w}\|_{\mathcal H^2}
 \leq\|H(w,T)\|_{\mathcal H^2}.
\]
More precisely, the integrand map is the GKW projection followed by the
inverse of the stochastic-integral isometry onto its closed range. As a
bounded linear map into $L^2(S)$, it preserves strong measurability.
Thus~\eqref{eq:fourier-stochastic-fubini} also makes the corresponding
$a$-weighted maps Bochner integrable in $L^2(S)$ and $\mathcal H^2$.
Commuting the two GKW projections with the Bochner integral now
yields~\eqref{eq:theta-fourier} and~\eqref{eq:L-fourier}. The strategy
identity is understood in $L^2(S)$ and hence up to $\mu_S$-null sets.
\end{proof}

For the error of a Fourier portfolio it is useful to introduce the cross
kernel
\begin{align}
 A_t^{w,u}
 &:=(1-\rho^2)w\overline u
  +(\rho w+\sqrt c\,\psi_t^w)
    \overline{(\rho u+\sqrt c\,\psi_t^u)}
\nonumber\\
&\quad+\int_{(0,\infty)}r_t^w(z)\overline{r_t^u(z)}\nu(\dd z),
 \label{eq:cross-A}\\
\Gamma_t^{w,u}
 &:=A_t^{w,u}-\frac{B_t^w\overline{B_t^u}}{D}.
 \label{eq:cross-Gamma}
\end{align}
Direct calculation from~\eqref{eq:residual-dynamics} gives
\[
 \dd\langle L^w,\overline{L^u}\rangle_t
 =H_{t-}(w,T)\overline{H_{t-}(u,T)}
   V_t\Gamma_t^{w,u}\dd t.
\]
The martingale isometry and Fubini's theorem for Bochner integrals therefore
give
\begin{align}
 \E^\Q[|L_T^Y|^2]
 &=\int_I\int_Ia(\lambda)\overline{a(\eta)}
 \E^\Q\left[
  \int_0^T
    H_{t-}(w(\lambda),T)
    \overline{H_{t-}(w(\eta),T)}
    V_t\Gamma_t^{w(\lambda),w(\eta)}\dd t
 \right]\dd\lambda\dd\eta.
\label{eq:error-general-fourier}
\end{align}
The double integral is absolutely convergent: the Cauchy--Schwarz inequality
for terminal residuals bounds its absolute integrand by
$|a(\lambda)a(\eta)|
\|H(w(\lambda),T)\|_{\mathcal H^2}
\|H(w(\eta),T)\|_{\mathcal H^2}$, which is integrable
by~\eqref{eq:fourier-stochastic-fubini}.
Equation~\eqref{eq:error-general-fourier} is a direct spectral representation
of the minimum error. An enlarged affine transform can further reduce its
expectation whenever the corresponding moment domain is available.

\subsection{European calls}

Let $\cK>0$ be a strike and define $w_\lambda=1/2+\mathrm{i}\lambda$. The
representation of \citet{lewis2001}, used in \citet{bondi2024}, gives the
time-$t$ call-price martingale
\begin{equation}
\label{eq:conditional-call}
 C_t(T,\cK)
 =S_t-\frac{\sqrt{S_0\cK}}{2\pi}
  \int_\R
   \frac{e^{\mathrm{i}\lambda\log(S_0/\cK)}H_t(w_\lambda,T)}
        {\lambda^2+1/4}\dd\lambda.
\end{equation}
Here the integral is an $L^2(\Q)$ Bochner integral. It is real because the
integrands at $\lambda$ and $-\lambda$ are complex conjugates. Indeed,
\[
 \|H(w_\lambda,T)\|_{\mathcal H^2}^2
 =\E^\Q[|e^{w_\lambda X_T}|^2]
 =\E^\Q[e^{X_T}]=1,
\]
and
\[
 \int_\R\frac{\dd\lambda}{\lambda^2+1/4}<\infty.
\]
The map $\lambda\mapsto e^{w_\lambda X_T}$ is continuous in $L^2(\Q)$
by dominated convergence, since squared differences are bounded by
$4e^{X_T}$. In particular, the required strong measurability holds.
Consequently~\eqref{eq:fourier-stochastic-fubini} holds automatically for the
call weight. The call itself belongs to $L^2$ because
$0\leq(S_T-\cK)^+\leq S_T$ and $S_T\in L^2(\Q)$.

Applying Theorem~\ref{thm:fourier-hedge} to this Lewis representation gives
the call hedge.

\begin{corollary}[Variance-optimal hedge for a European call]
\label{cor:call-hedge}
Suppose Assumptions~\ref{ass:model}, \ref{ass:transform}, and~\ref{ass:l2}
hold. The variance-optimal number of shares for the payoff $(S_T-\cK)^+$ is
\begin{equation}
 \vartheta_t^{\mathrm{call}}
 =1-\frac{\sqrt{S_0\cK}}{2\pi}
  \int_\R
   \frac{e^{\mathrm{i}\lambda\log(S_0/\cK)}\vartheta_t^{w_\lambda}}
        {\lambda^2+1/4}\dd\lambda,
\label{eq:call-hedge}
\end{equation}
where $\vartheta^{w_\lambda}$ is given by~\eqref{eq:theta-explicit}. The
residual hedging martingale is
\begin{equation}
\label{eq:call-residual}
 L_t^{\mathrm{call}}
 =-\frac{\sqrt{S_0\cK}}{2\pi}
  \int_\R
   \frac{e^{\mathrm{i}\lambda\log(S_0/\cK)}L_t^{w_\lambda}}
        {\lambda^2+1/4}\dd\lambda.
\end{equation}
Its variance is obtained by inserting the call Fourier weight
into~\eqref{eq:error-general-fourier}.
\end{corollary}

\begin{proof}
The terminal Lewis identity holds in $L^2(\Q)$ by the integrability argument
above. The claim $S_T$ has GKW integrand equal to one and zero residual.
Applying Theorem~\ref{thm:fourier-hedge} to the remaining Fourier integral and
using linearity of the GKW projection
gives~\eqref{eq:call-hedge}--\eqref{eq:call-residual}.
\end{proof}

\begin{remark}[Exponential-mark specialization]
The calibrated specification of \citet{bondi2024} uses
$\nu(\dd z)=e^{-z}\dd z$. In this case the jump integrals in the strategy are
available in closed form. Let $a_t^w=\psi_t^w-\Lambda w$. Then
\begin{align}
 D
 &=2+\frac{1}{1+2\Lambda}-\frac{2}{1+\Lambda},
 \label{eq:D-exp-marks}\\
 J(a_t^w)
 &:=\int_0^\infty(e^{a_t^wz}-1)(e^{-\Lambda z}-1)e^{-z}\dd z
 \nonumber\\
 &=\frac{1}{1-a_t^w+\Lambda}
   -\frac{1}{1-a_t^w}
   -\frac{1}{1+\Lambda}+1,
 \label{eq:J-exp-marks}
\end{align}
provided the real parts of the denominators are positive. Therefore
\begin{equation}
\label{eq:beta-exp-marks}
 \beta_t^w=\frac{w+\rho\sqrt c\,\psi_t^w+J(a_t^w)}{D}.
\end{equation}
This reduces evaluation of the hedge to the Fourier integral and the
Riccati--Volterra solve.
\end{remark}

\section{Convergence of approximate hedges in the original market}
\label{sec:original-market-stability}

We now establish convergence of predictable kernel-regularized hedges in the
original rough Hawkes--Heston market. The analysis combines stability of
the Riccati equation, predictable reconstruction of the singular memory,
and weighted control of the common-jump exposures. Throughout this section,
the probability space, the stock $S$, the variance $V$, the driver $Z$, and
the initial curve $g_0$ are fixed. On this market, we regularize the
deterministic kernel in the hedge formula and increase the Fourier cutoff.
The holdings are defined by exact Riccati solutions, history integrals, and
Fourier integrals using information available strictly before the trading
time. Their gains are continuous-time stochastic integrals against $S$.

\begin{assumption}[Kernel regularization]
\label{ass:original-kernel-regularization}
For every $n$, the kernel $K_n$ satisfies the kernel conditions in
Assumption~\ref{ass:model}, and
\[
 \epsilon_n:=\|K_n-K\|_{L^1(0,T)}\longrightarrow0.
\]
\end{assumption}

The integrated first-moment bound~\eqref{eq:original-integrated-first-moment}
follows from the standing model assumptions and will control the stochastic
history terms. Write
$w_\lambda=1/2+\mathrm{i}\lambda$ and let $\psi_w^n$ solve
\[
 \psi_w^n=K_n*\Rr(w,\psi_w^n).
\]
The Riccati existence result recalled in Section~\ref{sec:model} applies to
each $K_n$ and gives a global continuous solution with
$\operatorname{Re}\psi_w^n\leq0$.

The first step is to control both the Riccati solutions and the induced
memory terms uniformly on compact Fourier intervals.

\begin{lemma}[Riccati and memory approximation]
\label{lem:original-deterministic-stability}
For every $R<\infty$, there is a constant $C_R<\infty$ such that, for all
sufficiently large $n$,
\begin{equation}
\label{eq:original-riccati-stability}
 \sup_{|\lambda|\leq R}\sup_{0\leq s\leq T}
 |\psi_{w_\lambda}^n(s)-\psi_{w_\lambda}(s)|\leq C_R\epsilon_n.
\end{equation}
The solutions are uniformly bounded on this compact Fourier interval.
Set $F_n^w(s)=\Rr(w,\psi_w^n(s))$ and $F^w(s)=\Rr(w,\psi_w(s))$, and define
\begin{align}
 a_n^w(t,r)&=\int_t^T F_n^w(T-s)K_n(s-r)\dd s,
       &&0\leq r\leq t\leq T,\label{eq:original-memory-coefficient}\\
 U_n^w(t)&=\int_t^T F_n^w(T-s)g_0(s)\dd s
       +\int_{[0,t)}a_n^w(t,r)\dd Z_r.
       \label{eq:original-log-memory}
\end{align}
Let $a^w,U^w$ denote the same expressions with $K,\psi_w$. Then
\begin{equation}
\label{eq:original-memory-stability}
 \sup_{|\lambda|\leq R}\sup_{r\leq t}|a_n^{w_\lambda}(t,r)-a^{w_\lambda}(t,r)|
 \leq C_R\epsilon_n,
 \qquad
 \sup_{|\lambda|\leq R}\E^\Q\!\int_0^T
 |U_n^{w_\lambda}(t)-U^{w_\lambda}(t)|\dd t
 \leq C_R\epsilon_n.
\end{equation}
\end{lemma}

The predictable version of~\eqref{eq:original-log-memory} can be constructed
explicitly. Define the history statistic
\[
 \widehat V_n(t)=g_0(t)+\int_{[0,t)}K_n(t-r)\dd Z_r.
\]
Its martingale-convolution part belongs to $L^2(\Q\otimes\dd t)$ and its
drift-convolution part belongs to $L^1(\Q\otimes\dd t)$, by
$K_n\in L^2(0,T)$ and~\eqref{eq:original-integrated-first-moment}.
In particular, $\E^\Q\int_0^T|\widehat V_n(t)|\dd t<\infty$;
as an argument of the conditional transform, this history statistic may take
either sign.
The process
\begin{equation}
\label{eq:original-memory-semimartingale}
 \mathscr U_{n,t}^w=U_n^w(0)
 +\int_0^t\psi_w^n(T-s)\dd Z_s
 -\int_0^t F_n^w(T-s)\widehat V_n(s)\dd s
\end{equation}
 is therefore a c\`adl\`ag adapted semimartingale. The stochastic Fubini theorem identifies
it with the right-endpoint version of~\eqref{eq:original-log-memory}:
its diagonal coefficient is $a_n^w(t,t)=\psi_w^n(T-t)$ and its time
derivative off the diagonal is $-F_n^w(T-t)K_n(t-r)$.
Thus $U_n^w(t)=\mathscr U_{n,t-}^w$, with
$\mathscr U_{n,0-}^w=\mathscr U_{n,0}^w$, gives the desired predictable version.
The same construction applies to $U^w$ using $K$. Parameter-dependent
stochastic integration gives versions jointly measurable with respect to
the predictable sigma-field and the Borel sigma-field in $\lambda$.
 In particular, the stochastic Fubini theorem and~\eqref{eq:H-affine} give
\begin{equation}
\label{eq:original-predictable-transform}
 H_{t-}(w,T)=\exp\{wX_{t-}+U^w(t)\},
 \qquad \Q\otimes\dd t\text{-almost everywhere}.
\end{equation}

\begin{proof}[Proof of Lemma~\ref{lem:original-deterministic-stability}]
The map $\Rr(w,v)$ is locally Lipschitz in $v$, uniformly for $w$ on a
compact part of the Lewis contour and $v$ in a bounded part of the closed
left half-plane. Indeed, differentiation of its jump term is justified by
$\int_{(0,\infty)}z^2\nu(\dd z)<\infty$ and
$|e^{az}-1|\leq |a|z$ for $\operatorname{Re}a\leq0$.
Continuous dependence on $w$ follows by the same local Volterra argument.
The global original solutions therefore form a bounded family on
$[0,T]\times[-R,R]$.

Until an approximating solution leaves a fixed larger bounded set,
subtraction of the two Riccati equations gives
\[
 |\psi_w^n-\psi_w|(s)
 \leq C_R\epsilon_n+L_R(K_n*|\psi_w^n-\psi_w|)(s).
\]
Choose $h>0$ so that $L_R\int_0^h K(s)\dd s<1/4$. Convergence in
$L^1$ gives $L_R\int_0^h K_n(s)\dd s<1/2$ for all sufficiently large $n$.
Taking suprema on successive intervals of length $h$, and using the
uniform bound on $\|K_n\|_{L^1(0,T)}$ for the preceding intervals,
yields~\eqref{eq:original-riccati-stability}. Its right-hand side tends to zero,
so the exit cannot occur. This also proves the asserted uniform bounds
and $\sup_{w,s}|F_n^w(s)-F^w(s)|\leq C_R\epsilon_n$.

Consequently,
\[
 |a_n^w-a^w|
 \leq \|F_n^w-F^w\|_\infty\|K_n\|_{L^1(0,T)}
       +\|F^w\|_\infty\epsilon_n
 \leq C_R\epsilon_n.
\]
To compare the stochastic terms, write
$Z_t=b\int_0^t V_r\dd r+M_t$, where $M$ is square-integrable and
\[
 \dd\langle M\rangle_r=(c+m_2)V_r\dd r,
 \qquad m_2=\int_{(0,\infty)}z^2\nu(\dd z).
\]
With $\delta a=a_n^w-a^w$, the two estimates
\begin{align*}
 \E^\Q\left|b\int_0^t\delta a(t,r)V_r\dd r\right|
 &\leq |b|\|\delta a\|_\infty\E^\Q\int_0^T V_r\dd r,\\
 \E^\Q\left|\int_{[0,t)}\delta a(t,r)\dd M_r\right|^2
 &\leq(c+m_2)\|\delta a\|_\infty^2\E^\Q\int_0^T V_r\dd r
\end{align*}
prove~\eqref{eq:original-memory-stability}, including the deterministic
$g_0$ term. Only the first integrated moment of $V$ is used.
For the history statistic, the corresponding time-integrated bounds are
$|b|\|K_n\|_{L^1}\E^\Q\int_0^T V_r\dd r$ for the absolute drift convolution
and $(c+m_2)\|K_n\|_{L^2}^2\E^\Q\int_0^T V_r\dd r$ for the squared martingale
convolution.
They justify the construction~\eqref{eq:original-memory-semimartingale}.
 For the stochastic Fubini theorem, the square-integrable kernel condition gives, for
each fixed $t$,
\[
 \int_t^T\!\left(\E^\Q\int_0^t K(s-r)^2V_r\dd r\right)^{1/2}\dd s
 \leq\sqrt T\,\|K\|_{L^2(0,T)}
       \left(\E^\Q\int_0^T V_r\dd r\right)^{1/2}<\infty,
\]
and likewise for $K_n$. The same bound on the triangle
$0\leq r\leq s\leq t$ justifies~\eqref{eq:original-memory-semimartingale}.
The finite-variation terms are absolutely
integrable by the $L^1$ kernel bound. This also
establishes~\eqref{eq:original-predictable-transform} in the stated sense.
\end{proof}

\paragraph{Predictable approximate exposures.}
Applying the conditional Cauchy--Schwarz inequality with respect to
$\mathcal F_{t-}$, together with the martingale property of $S$, gives
\begin{equation}
\label{eq:original-transform-envelope}
 |H_{t-}(w_\lambda,T)|^2\leq S_{t-}/S_0.
\end{equation}
Hence $\operatorname{Re}U^{w_\lambda}(t)\leq0$ almost everywhere.
Define
\begin{equation}
\label{eq:original-clipped-transform}
 \chi(z)=\exp\{\min(\operatorname{Re}z,0)+\mathrm{i}\operatorname{Im}z\},
 \qquad
 \widetilde H_n^w(t)=e^{wX_{t-}}\chi(U_n^w(t)).
\end{equation}
This truncation projects $e^{wX_{t-}+U_n^w(t)}$ onto the closed disk
of radius $\sqrt{S_{t-}/S_0}$. In particular,
\begin{equation}
\label{eq:original-approx-envelope}
 |\chi(U_n^w(t))|\leq1,
 \qquad
 |\widetilde H_n^w(t)|^2\leq S_{t-}/S_0.
\end{equation}
The bounded factor $\chi(U_n^w)$ ensures that the approximate transform
preserves the exact transform's random modulus envelope.
With $D$ and $q$ as in~\eqref{eq:D-def}, put
\begin{align}
 \beta_n^w(t)&=\frac{1}{D}\left[
 w+\rho\sqrt c\,\psi_w^n(T-t)
 +\int_{(0,\infty)}\left(e^{(\psi_w^n(T-t)-\Lambda w)z}-1\right)q(z)\nu(\dd z)
 \right],\label{eq:original-approx-beta}\\
 \widetilde\vartheta_n^w(t)&=
 \frac{\widetilde H_n^w(t)}{S_{t-}}\beta_n^w(t).
 \label{eq:original-approx-holding}
\end{align}

We now transfer the Riccati and memory stability estimates to these
predictable approximate holdings.

\begin{theorem}[Hedge convergence on compact Fourier intervals]
\label{thm:original-compact-hedge}
Suppose Assumptions~\ref{ass:model}, \ref{ass:transform}, \ref{ass:l2},
and~\ref{ass:original-kernel-regularization} hold. For each $R<\infty$,
the predictable strategies~\eqref{eq:original-approx-holding} belong to
$L^2(S)$ and
\begin{equation}
\label{eq:original-compact-convergence}
 \sup_{|\lambda|\leq R}
 \|\widetilde\vartheta_n^{w_\lambda}
       -\vartheta^{w_\lambda}\|_{L^2(S)}\longrightarrow0.
\end{equation}
Consequently, for every $a\in L^1([-R,R];\C)$, the corresponding Fourier
integrals of the approximate holdings converge in $L^2(S)$ to the
Fourier integral of the exact holdings.
\end{theorem}

\begin{proof}
By the stock bracket formula and square-integrability,
$\E^\Q\int_0^T S_{t-}^2V_t\dd t<\infty$. Thus
\begin{equation}
\label{eq:original-weight-integrability}
 \E^\Q\int_0^T S_{t-}V_t\dd t
 \leq\left(\E^\Q\int_0^T S_{t-}^2V_t\dd t\right)^{1/2}
      \left(\E^\Q\int_0^T V_t\dd t\right)^{1/2}<\infty.
\end{equation}
For each fixed $n$ and $w$, the deterministic coefficient $\beta_n^w$ is
bounded on $[0,T]$. Setting $W_t=S_{t-}V_t/S_0$, the envelope gives
\[
 \|\widetilde\vartheta_n^w\|_{L^2(S)}^2
 =D\E^\Q\int_0^T V_t|\beta_n^w(t)\widetilde H_n^w(t)|^2\dd t
 \leq D\|\beta_n^w\|_\infty^2\E^\Q\int_0^T W_t\dd t<\infty.
\]
Thus admissibility follows from weighted integrability of the random
envelope.

The map $\chi$ is globally $1$-Lipschitz and has modulus at most one:
on each open half-plane its real differential has operator norm at most
one, and $\chi$ is continuous across the imaginary axis. Integrating along
line segments therefore gives the Lipschitz bound.
Together with~\eqref{eq:original-transform-envelope}, this gives
\[
 |\widetilde H_n^w(t)-H_{t-}(w,T)|^2
 \leq\frac{S_{t-}}{S_0}
        \min\{4,|U_n^w(t)-U^w(t)|^2\}.
\]
For uniform convergence in $|\lambda|\leq R$, define the weight tail
$\mathcal T_W(B):=\E^\Q\int_0^T W_t\mathbf1_{\{W_t>B\}}\dd t$.
For every $B>0$,
\begin{align*}
 \E^\Q\int_0^T V_t|\widetilde H_n^w(t)-H_{t-}(w,T)|^2\dd t
 &\leq 4\mathcal T_W(B)
       +2B\E^\Q\int_0^T|U_n^w(t)-U^w(t)|\dd t.
\end{align*}
Lemma~\ref{lem:original-deterministic-stability} bounds the second term
uniformly by $C_R B\epsilon_n$, while
$\mathcal T_W(B)\to0$ as $B\to\infty$ by~\eqref{eq:original-weight-integrability}.

The deterministic coefficients $\beta_n^w$ are uniformly bounded on
compact Fourier intervals and
\[
 |\beta_n^w(t)-\beta_t^w|
 \leq\frac{|\rho|\sqrt c+\|q\|_{L^2(\nu)}\sqrt{m_2}}{D}
       |\psi_w^n(T-t)-\psi_w(T-t)|.
\]
Indeed, the difference of the two jump exponentials is bounded by
$z|\psi_w^n-\psi_w|$. Finally,
\begin{align*}
 \|\widetilde\vartheta_n^w-\vartheta^w\|_{L^2(S)}^2
 &=D\E^\Q\int_0^T V_t
       |\beta_n^w(t)\widetilde H_n^w(t)-\beta_t^wH_{t-}(w,T)|^2\dd t\\
 &\leq 2D\|\beta_n^w\|_\infty^2
       \E^\Q\int_0^T V_t|\widetilde H_n^w(t)-H_{t-}(w,T)|^2\dd t\\
 &\quad+2D\|\beta_n^w-\beta^w\|_\infty^2
       \E^\Q\int_0^T W_t\dd t.
\end{align*}
Consequently, for all sufficiently large $n$ and every $B>0$,
\begin{equation}
\label{eq:original-weighted-error-bound}
 \sup_{|\lambda|\leq R}
 \|\widetilde\vartheta_n^{w_\lambda}-\vartheta^{w_\lambda}\|_{L^2(S)}^2
 \leq C_R\bigl(\mathcal T_W(B)+B\epsilon_n+\epsilon_n^2\bigr),
\end{equation}
where $C_R$ is independent of $n$ and $B$. Taking $n\to\infty$ at fixed
$B$, then $B\to\infty$, proves~\eqref{eq:original-compact-convergence}.
For fixed $n$, continuous dependence of $\psi_w^n$ and
$F_n^w$ on $\lambda$, followed by the same memory and weighted truncation
estimates, gives continuity of
$\lambda\mapsto\widetilde\vartheta_n^{w_\lambda}$ in $L^2(S)$; the exact
holdings have the same property. This ensures strong measurability of the
Fourier integrands. Minkowski's inequality then
proves the final assertion.
\end{proof}

We next remove the Fourier cutoff and translate strategy convergence into
convergence of trading gains and hedging errors.

\begin{corollary}[Trading gains and mean-square error for a call]
\label{cor:original-call-stability}
Under the assumptions of Theorem~\ref{thm:original-compact-hedge}, set
\[
 a_{\cK}(\lambda)=\frac{\sqrt{S_0\cK}}{2\pi}
 \frac{e^{\mathrm{i}\lambda\log(S_0/\cK)}}{\lambda^2+1/4},
\]
and, for $R<\infty$, define
\begin{align}
 \vartheta_{n,R}(t)&=1-\int_{-R}^R
              a_{\cK}(\lambda)\widetilde\vartheta_n^{w_\lambda}(t)\dd\lambda,
              \label{eq:original-call-approx}\\
 x_{n,R}&=S_0-\int_{-R}^R
              a_{\cK}(\lambda)\chi(U_n^{w_\lambda}(0))\dd\lambda.
              \label{eq:original-call-capital}
\end{align}
These quantities are real by conjugate symmetry. There is a deterministic
sequence $R_n\uparrow\infty$ such that, writing
$\vartheta_n=\vartheta_{n,R_n}$ and $x_n=x_{n,R_n}$,
\[
 x_n\longrightarrow C_0(T,\cK),\qquad
 \|\vartheta_n-\vartheta^{\mathrm{call}}\|_{L^2(S)}\longrightarrow0.
\]
In particular, the original-market gain processes
$G_n(t):=\int_0^t\vartheta_n(s)\dd S_s$ and
$G^*(t):=\int_0^t\vartheta^{\mathrm{call}}(s)\dd S_s$, $0\leq t\leq T$,
satisfy $G_n\to G^*$ in $\mathcal H^2$, and
\begin{equation}
\label{eq:original-uniform-gains}
 \E^\Q\sup_{0\leq t\leq T}|G_n(t)-G^*(t)|^2
 \leq4\|\vartheta_n-\vartheta^{\mathrm{call}}\|_{L^2(S)}^2\longrightarrow0.
\end{equation}
For $Y=(S_T-\cK)^+$ and
$\varepsilon_*^2=\E^\Q[(L_T^{\mathrm{call}})^2]$, one has the exact identity
\begin{equation}
\label{eq:original-excess-mse}
 \E^\Q[(Y-x_n-G_n(T))^2]
 =\varepsilon_*^2+(x_n-C_0(T,\cK))^2
       +\|\vartheta_n-\vartheta^{\mathrm{call}}\|_{L^2(S)}^2
 \longrightarrow\varepsilon_*^2.
\end{equation}
\end{corollary}

\begin{proof}
Let $(x_R,\vartheta_R)$ be the exact truncated pair obtained by replacing
$\chi(U_n^{w_\lambda}(0))$ and $\widetilde\vartheta_n^{w_\lambda}$ in
\eqref{eq:original-call-capital} and~\eqref{eq:original-call-approx} with
$H_0(w_\lambda,T)$ and $\vartheta^{w_\lambda}$, respectively.
The initial memory is deterministic, so the Riccati estimate gives
\[
 \sup_{|\lambda|\leq R}|U_n^{w_\lambda}(0)-U^{w_\lambda}(0)|
 \leq\|g_0\|_{L^1(0,T)}
       \sup_{|\lambda|\leq R}\|F_n^{w_\lambda}-F^{w_\lambda}\|_\infty
 \leq C_R\epsilon_n.
\]
Since $X_0=0$ and $|H_0(w_\lambda,T)|\leq1$, the affine transform gives
$\chi(U^{w_\lambda}(0))=H_0(w_\lambda,T)$.
The Lipschitz property of $\chi$ and integrability of $a_{\cK}$ therefore
give capital convergence at each fixed cutoff. Together with
Theorem~\ref{thm:original-compact-hedge}, this proves
\[
 d_n(R):=|x_{n,R}-x_R|
       +\|\vartheta_{n,R}-\vartheta_R\|_{L^2(S)}\longrightarrow0
 \qquad\text{for each fixed }R>0.
\]
On the Lewis contour,
$\|\vartheta^{w_\lambda}\|_{L^2(S)}\leq
\|e^{w_\lambda X_T}\|_{L^2(\Q)}=1$ by GKW contraction, and
$|H_0(w_\lambda,T)|\leq1$. The omitted tails of the exact hedge and
capital are each bounded by
\[
 \tau(R):=\int_{|\lambda|>R}|a_{\cK}(\lambda)|\dd\lambda
 \leq\frac{\sqrt{S_0\cK}}{\pi R},\qquad R>0.
\]
Hence
\begin{equation}
\label{eq:original-cutoff-error-bound}
 |x_{n,R}-C_0(T,\cK)|
 +\|\vartheta_{n,R}-\vartheta^{\mathrm{call}}\|_{L^2(S)}
 \leq d_n(R)+2\tau(R).
\end{equation}
Choose strictly increasing integers $N_m\geq m$ such that
$d_n(m)\leq1/m$ for every $n\geq N_m$. This is possible by fixed-cutoff
convergence. For $n\geq N_1$, set
\[
 R_n=\max\{m\in\mathbb N:N_m\leq n\},
\]
and put $R_n=1$ for the finitely many indices $n<N_1$.
Then $R_n$ is nondecreasing, $R_n\to\infty$, and
$d_n(R_n)+2\tau(R_n)\leq R_n^{-1}+2\tau(R_n)\to0$.
Equation~\eqref{eq:original-cutoff-error-bound} proves both asserted
convergences for the cutoff sequence constructed above.
The isometry and Doob's inequality give~\eqref{eq:original-uniform-gains}.
Finally, substitute the call's GKW decomposition into the left-hand side
of~\eqref{eq:original-excess-mse}. The capital difference is orthogonal to
all zero-mean martingale gains, and $L^{\mathrm{call}}$ is orthogonal to
every stock integral. All cross terms vanish, proving the identity.
\end{proof}

\begin{remark}[Regularizing the fractional kernel]
\label{rem:original-shifted-kernel}
For $K=K_\alpha$ and $\alpha\in(1/2,1)$, one concrete choice is
\[
 K_\varepsilon(t)=\frac{(t+\varepsilon)^{\alpha-1}}{\Gamma(\alpha)},
 \qquad\varepsilon\downarrow0.
\]
Each shifted kernel is completely monotone, belongs to
$L^2_{\mathrm{loc}}$, and satisfies the nonnegative, nonincreasing
resolvent conditions; see \citet[Example~2.4]{abijaber2021weak} and
\citet{gripenberg1990}. Direct integration gives
\[
 \|K_\varepsilon-K_\alpha\|_{L^1(0,T)}
 =\frac{\varepsilon^\alpha+T^\alpha-(T+\varepsilon)^\alpha}
        {\Gamma(\alpha+1)}
 \leq\frac{\varepsilon^\alpha}{\Gamma(\alpha+1)}.
\]
Thus the preceding results apply while retaining the original rough
market and its original initial curve. The bound above quantifies the
deterministic kernel perturbation, while
Corollary~\ref{cor:original-call-stability} combines it with the weighted
transform argument and the diagonal choice of Fourier cutoffs to obtain the
call-hedging convergence.
\end{remark}

\section{Numerical illustration in the original rough market}
\label{sec:original-market-numerics}

This section evaluates the approximation along the same axis as
Theorem~\ref{thm:original-compact-hedge}: one common ensemble of original
rough-market paths is held fixed while the deterministic kernel in the hedge
formula is shifted.
The organization follows two ideas from the approximation literature. First,
\citet{bayerbreneis2023weak} formulate weak rough-Heston pricing error in
terms of the $L^1$ kernel error, the same perturbation used in
Assumption~\ref{ass:original-kernel-regularization}. Second,
\citet{dinunnoyurchenko2026} use common Brownian realizations when comparing
Volterra hedge approximations. We adapt this common-random-number coupling to
evaluate every regularized holding on one fixed numerical market.

\subsection{Fixed-market design}

We use the exponential-mark calibration in \citet{bondi2024}, namely
\[
 (\alpha,\rho,b,c,\Lambda,\beta,V_0)
 =(0.527,-0.731,-1.812,0.115,0.276,0.049,0.0079),
\]
with $S_0=1$, $g_0(t)=V_0+\beta t^\alpha/\Gamma(\alpha+1)$,
$T=0.25$, and an at-the-money call. The original fractional market is
simulated once on a grid $t_k=kh$, $h=T/M$. To handle the rough square-root
state, we use the integrated-variance Euler construction of
\citet{richardtan2023} and its running-maximum clock. Writing
$A_t=\int_0^tV_s\dd s$ and
$G_0(t)=\int_0^tg_0(s)\dd s$, the grid recursion begins with
\begin{equation}
\label{eq:numerical-integrated-variance}
 A_{k+1}^h=G_0(t_{k+1})
   +h\sum_{j=0}^{k}K_\alpha(t_{k+1}-t_j)Z_j^h,
 \qquad
 \overline A_{k+1}^h=\max_{0\leq j\leq k+1}A_j^h.
\end{equation}
The Brownian martingales are advanced with conditional variance
$\Delta\overline A_{k+1}^h$. The common marked jumps are generated with the
same clock: unit-rate Poisson thresholds falling in
$(\overline A_k^h,\overline A_{k+1}^h]$ receive independent exponential
marks. This gives the discrete compensated driver
\[
 Z_{k+1}^h=(b-1)\overline A_{k+1}^h
       +\sqrt c\,M_{2,k+1}^h+J_{k+1}^h,
\]
where $J^h$ is the cumulative mark sum. The log stock is updated from the
same Brownian and marked-jump variables according to
\[
 X_k^h=(-\tfrac12-\bar q)\overline A_k^h
 +\sqrt{1-\rho^2}\,M_{1,k}^h+\rho M_{2,k}^h-\Lambda J_k^h,
 \qquad \bar q=\int_{(0,\infty)}q(z)\nu(\dd z),
\]
so the variance and stock jumps are common as in~\eqref{eq:variance-driver} and~\eqref{eq:stock-martingale}.

For every
\[
 \frac{\varepsilon}{T}\in
 \left\{\frac14,\frac18,\frac1{16},\frac1{32},\frac1{64}\right\},
 \qquad
 K_\varepsilon(t)=\frac{(t+\varepsilon)^{\alpha-1}}
                         {\Gamma(\alpha)},
\]
we solve the corresponding Riccati--Volterra equation by product integration.
The fixed original driver generates every integrated approximate history:
\begin{equation}
\label{eq:numerical-fixed-history}
 A_{\varepsilon,k}^h=G_0(t_k)
 +h\sum_{j=0}^{k-1}K_\varepsilon(t_k-t_j)Z_j^h.
\end{equation}
The predictable log memory is then advanced by the discrete counterpart
of~\eqref{eq:original-memory-semimartingale},
\[
 U_{\varepsilon,k+1}^w-U_{\varepsilon,k}^w
 =\psi_w^\varepsilon(T-t_k)\Delta Z_{k+1}^h
  -F_\varepsilon^w(T-t_k)\Delta A_{\varepsilon,k+1}^h,
\]
and is passed through the map $\chi$ in~\eqref{eq:original-clipped-transform}. Thus every
$\vartheta_{\varepsilon,R}^h(t_k)$ trades the same stock path over
$(t_k,t_{k+1}]$ and is measurable with respect to the simulation history
$\F_k^h$ through $t_k$, including $Z_0^h,\ldots,Z_k^h$.

The principal statistic is a direct quadrature of the $L^2(S)$ norm. For
$P$ paths, define
\begin{equation}
\label{eq:numerical-strategy-norm}
 \widehat d_{\varepsilon,R}^2
 =\frac1P\sum_{p=1}^P\sum_{k=0}^{M-1}
 D(S_k^{h,p})^2
 \left|\vartheta_{\varepsilon,R}^{h,p}(t_k)
             -\vartheta_{0,R}^{h,p}(t_k)\right|^2
 \Delta\overline A_{k+1}^{h,p}.
\end{equation}
The predictable bracket
$\dd\langle S\rangle=D S_{-}^2\dd A$ provides the quadrature weight
in~\eqref{eq:numerical-strategy-norm}. At a fixed grid size, the simulated
stock has a different exact bracket. Since
$a_k:=\Delta\overline A_{k+1}^h$ is $\F_k^h$-measurable, the conditional
Gaussian and compound-Poisson moment formulas give
\begin{equation}
\label{eq:numerical-discrete-bracket}
 \E^\Q[\Delta S_{k+1}^h\mid\F_k^h]=0,
 \qquad
 \E^\Q[(\Delta S_{k+1}^h)^2\mid\F_k^h]
 =(S_k^h)^2\left(e^{Da_k}-1\right).
\end{equation}
Thus $D(S_k^h)^2a_k$ is the first-order approximation to the discrete
bracket weight, since $e^{Da_k}-1=Da_k+O(a_k^2)$ as $a_k\to0$.
Equation~\eqref{eq:numerical-strategy-norm} estimates the continuous-market
strategy norm by quadrature; the exact discrete-time gain isometry uses
the exponential weight in~\eqref{eq:numerical-discrete-bracket}.
We also compute
\[
 \widehat g_{\varepsilon,R}^2
 =\frac1P\sum_{p=1}^P\max_{0\leq k\leq M}
 \left|\sum_{j<k}
  (\vartheta_{\varepsilon,R}^{h,p}(t_j)
   -\vartheta_{0,R}^{h,p}(t_j))\Delta S_{j+1}^{h,p}\right|^2.
\]
Both statistics are paired path by path. Their Monte Carlo standard errors
are computed from the corresponding pathwise quantities.

The production run uses $P=5{,}000$, $M=512$, composite Gauss--Legendre
quadrature, and Fourier cutoffs $R\in\{5,10,20,40,80\}$. At the smallest
shift, $\varepsilon=8h$. The common Brownian increments and Poisson
time-change variables are also aggregated to $M=256$ for an outer-grid check.

\subsection{Convergence results}

Table~\ref{tab:fixed-cutoff-convergence} fixes $R=20$, as in
Theorem~\ref{thm:original-compact-hedge}. The Riccati error is the supremum
over the time grid and all Fourier nodes with $|\lambda|\leq20$. The memory
column estimates the corresponding supremum over Fourier nodes of
$\E\int_0^T|U_\varepsilon^w(t)-U^w(t)|\dd t$. All five diagnostics decrease
monotonically with the exact kernel error from
Remark~\ref{rem:original-shifted-kernel}. In particular, the squared strategy
distance falls from $2.447\times10^{-6}$ to $2.170\times10^{-7}$, while the
uniform gain distance falls from $3.587\times10^{-6}$ to
$3.306\times10^{-7}$.

\begin{table}[H]
\centering
\caption{Original-market convergence at the fixed cutoff $R=20$.
Parentheses contain one Monte Carlo standard error.}
\label{tab:fixed-cutoff-convergence}
\small
\begin{tabular}{@{}rrrrrr@{}}
\toprule
$\varepsilon/T$ & $\|K_\varepsilon-K\|_1$ & Riccati & Memory
& $10^6\widehat d_{\varepsilon,20}^2$
& $10^6\widehat g_{\varepsilon,20}^2$ \\
\midrule
$1/4$  & 0.193673 & 11.305 & 0.035392 & 2.447 (0.173) & 3.587 (0.317) \\
$1/8$  & 0.146656 &  9.614 & 0.025336 & 1.411 (0.112) & 2.091 (0.221) \\
$1/16$ & 0.108279 &  8.014 & 0.017928 & 0.781 (0.065) & 1.172 (0.138) \\
$1/32$ & 0.078500 &  6.549 & 0.012658 & 0.417 (0.034) & 0.633 (0.078) \\
$1/64$ & 0.056184 &  5.251 & 0.008963 & 0.217 (0.017) & 0.331 (0.041) \\
\bottomrule
\end{tabular}
\end{table}

For the call diagonal, take the unshifted-kernel pair
$(x_{0,80},\vartheta_{0,80})$ as the numerical benchmark and measure the
combined squared capital and strategy distance to that benchmark:
\begin{equation}
\label{eq:numerical-reference-distance}
 \widehat\Delta_{\varepsilon,R}
 =(x_{\varepsilon,R}-x_{0,80})^2
 +\widehat{\|\vartheta_{\varepsilon,R}
                   -\vartheta_{0,80}\|}_{L^2(S)}^2.
\end{equation}
The strategy component uses the bracket quadrature
in~\eqref{eq:numerical-strategy-norm}. Thus the statistic measures convergence
to the benchmark at the specified Fourier cutoff and time grid. Its relation
to mean-square hedging error follows from a general benchmark identity.
For any admissible benchmark $(x_b,\vartheta_b)$ in the continuous market,
define $G_b(t):=\int_0^t\vartheta_b(s)\dd S_s$ and
$G(t):=\int_0^t\vartheta(s)\dd S_s$, $0\leq t\leq T$, and set
$e_b=Y-x_b-G_b(T)$ and
$\delta W=(x-x_b)+G(T)-G_b(T)$. Then
\begin{equation}
\label{eq:benchmark-mse-difference}
\begin{aligned}
 &\E^\Q[(Y-x-G(T))^2]-\E^\Q[e_b^2]\\
 &\qquad=(x-x_b)^2+\|\vartheta-\vartheta_b\|_{L^2(S)}^2
          -2\E^\Q[e_b\delta W].
\end{aligned}
\end{equation}
For the exact optimal benchmark, GKW orthogonality makes the cross term
vanish, yielding~\eqref{eq:original-excess-mse}.

The cutoff is increased along
$(\varepsilon/T,R)=(1/4,5),(1/8,10),(1/16,20),(1/32,40),(1/64,80)$.
Table~\ref{tab:fourier-diagonal-convergence} shows that both components
decrease, and their sum falls by more than three orders of magnitude.
These five pairs illustrate simultaneous kernel and Fourier refinement on
the fixed numerical market, with both distances evaluated against the
unshifted-kernel benchmark at $R=80$.

\begin{table}[H]
\centering
\caption{Fourier diagonal: squared capital and strategy distance to the
unshifted-kernel numerical benchmark at $R=80$.}
\label{tab:fourier-diagonal-convergence}
\small
\begin{tabular}{@{}rrrrr@{}}
\toprule
$\varepsilon/T$ & $R$ & Capital component & Strategy component
& $\widehat\Delta_{\varepsilon,R}$ \\
\midrule
$1/4$  &  5 & $2.256\times10^{-3}$ & $1.766\times10^{-4}$ & $2.432\times10^{-3}$ \\
$1/8$  & 10 & $3.359\times10^{-4}$ & $8.129\times10^{-5}$ & $4.172\times10^{-4}$ \\
$1/16$ & 20 & $3.495\times10^{-5}$ & $2.673\times10^{-5}$ & $6.168\times10^{-5}$ \\
$1/32$ & 40 & $3.602\times10^{-6}$ & $5.163\times10^{-6}$ & $8.765\times10^{-6}$ \\
$1/64$ & 80 & $1.516\times10^{-6}$ & $6.197\times10^{-7}$ & $2.136\times10^{-6}$ \\
\bottomrule
\end{tabular}
\end{table}

Figure~\ref{fig:original-market-stability} summarizes the proof-level
diagnostics, the fixed-cutoff convergence, and the call diagonal together.

\begin{figure}[H]
\centering
\includegraphics[width=\textwidth]{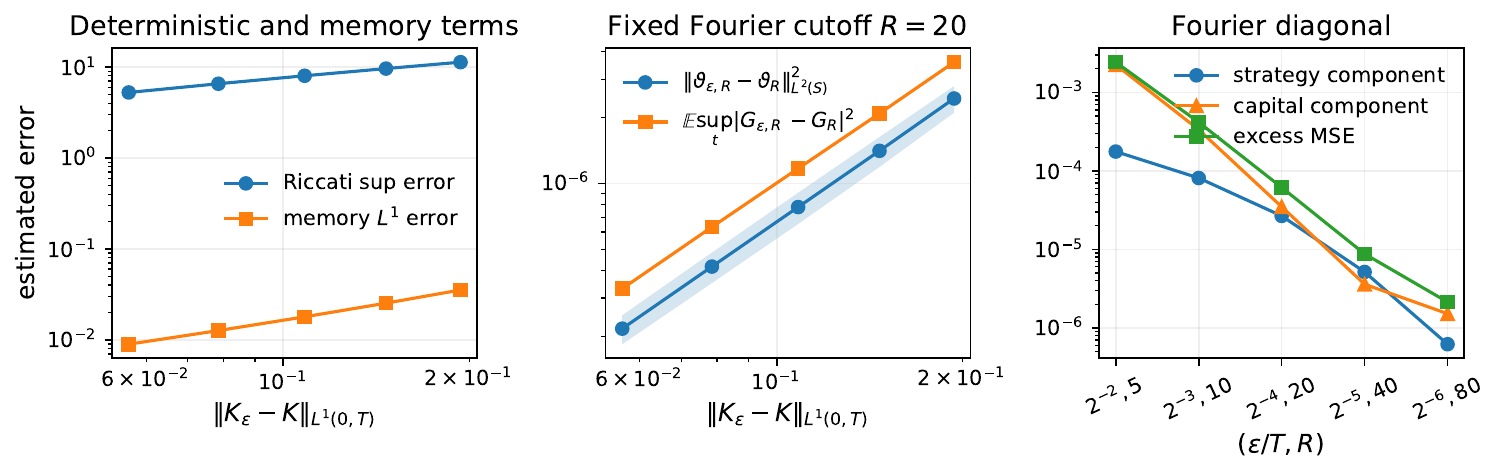}
\caption{Original-market kernel regularization. Left: deterministic Riccati
and pathwise memory diagnostics at $R=20$. Middle: the strategy and gain
distances at the same fixed cutoff; the shaded band is a pointwise 95\%
Monte Carlo interval for the strategy distance. Right: capital, strategy,
and total squared-distance components relative to the numerical benchmark
along the Fourier diagonal.}
\label{fig:original-market-stability}
\end{figure}

The outer-grid comparison gives
$\widehat d_{\varepsilon,20}^2=3.867\times10^{-7}$ at $M=256$ and
$4.171\times10^{-7}$ at $M=512$ for $\varepsilon/T=1/32$; the associated
gain statistics are $5.446\times10^{-7}$ and $6.327\times10^{-7}$.
The grid refinements change the two statistics by $7.9\%$ and $16.2\%$,
respectively. The fine-grid stock mean is $1.00078$ with standard error
$7.86\times10^{-4}$. Increasing the Gauss--Legendre order from four to six per
panel changes the five unshifted initial capitals by at most
$3.1\times10^{-6}$. These comparisons measure sensitivity to the tested
grid and quadrature refinements.

\section{Conclusion}
\label{sec:conclusion}

We have derived the variance-optimal stock hedge in the rough
Hawkes--Heston model and constructed approximations whose performance
converges in that same market. The hedge follows from the model's affine
conditional martingales and the classical GKW projection. Its explicit
Brownian and common marked-jump covariance terms identify the optimal
position, while the orthogonal residual gives the minimum quadratic error.
Fourier synthesis makes both quantities available for European calls.

The main approximation result connects regularization of the hedge formula
to convergence of trading performance with the original stock and
information flow fixed. Its proof combines predictable reconstruction of
the singular Volterra history, preservation of the transform's random
modulus envelope, and truncation of the integrable weight $S_{t-}V_t$.
Together with Riccati stability, these estimates yield convergence of the
holdings in $L^2(S)$ on compact Fourier intervals. A diagonal choice of
kernel regularization and Fourier cutoff then gives convergence of call
capitals and strategies, uniform-in-time square-mean convergence of
continuous-time gains, and convergence of terminal mean-square errors to
the original market's variance-optimal value.

Shifted fractional kernels provide a concrete family covered by the
construction. The numerical experiment illustrates the successive kernel,
memory, and trading approximations on common original-market paths. The
resulting analysis links an explicit covariance projection to admissible
approximate hedges and their limiting quadratic performance in a market
with rough memory and common jumps. Extensions to historical-measure
hedging, discrete rebalancing, and static option positions provide natural
directions for further work.

\section*{Acknowledgments}
Xiaoyu Wang is supported by the Guangzhou-HKUST(GZ) Joint Funding Program\linebreak (No.2025A03J3556).

\bibliographystyle{abbrvnat}
\bibliography{bibtex}

\end{document}